\documentclass[10pt,journal]{IEEEtran}
\usepackage{xcolor} 
\usepackage[colorlinks=true,linkcolor=black,urlcolor=blue,citecolor=blue]{hyperref}
\usepackage{amsmath,amssymb,amsfonts}
\usepackage{amsthm}
\usepackage{algorithmic}
\usepackage{algorithm}
\usepackage{array}
\usepackage[caption=false,font=normalsize,labelfont=sf,textfont=sf]{subfig}
\usepackage{textcomp}
\usepackage{stfloats}
\usepackage{siunitx}
\usepackage{url}
\usepackage{verbatim}
\usepackage{graphicx}
\usepackage{cite}
\DeclareSIUnit\dBm{dBm}
\DeclareSIUnit\mph{mph}
\usepackage{xcolor} 
\DeclareSIUnit\dBi{dBi}
\DeclareSIUnit\dB{dB}
\usepackage{booktabs}
\usepackage{multirow}
\usepackage{subcaption}

\newtheorem{lemma}{Lemma}

\newtheorem{remark}{Remark}

\title{Toward Safe and Energy-Efficient 5G NR V2X Communications in Rural Environments}

\author{Zhanle Zhao, Son Dinh-Van,~\IEEEmembership{Member,~IEEE}, Yuen Kwan Mo, Siddartha Khastgir and \\Matthew  D. Higgins,~\IEEEmembership{Senior Member,~IEEE}
	\thanks{This work is supported in part by the State Scholarship Fund by China Scholarship Council (CSC) from the Ministry of Education of P. R. China; in part by the WMG Centre High Value Manufacturing Catapult, University of Warwick, Coventry, United Kingdom.}
	\thanks{The authors are with Warwick Manufacturing Group, University of Warwick, Coventry, United Kingdom. Emails: \{zhanle.zhao, son.v.dinh,
		tony.mo, s.khastgir.1, m.higgins\}@warwick.ac.uk.}}

\begin{document}
	\maketitle
	
\begin{abstract}
	Connected braking can reduce fatal collisions in connected and autonomous vehicles (CAVs) by using reliable, low-latency 5G New Radio (NR) links, especially NR Sidelink Vehicle-to-Everything (V2X). In rural areas, road side units are sparse and power-constrained, so energy efficiency must be considered alongside safety. This paper studies how three communication control factors including subcarrier spacing ($\mathrm{SCS}$), modulation and coding scheme ($\mathrm{MCS}$), and transmit power ($P_{\mathrm{t}}$) should be configured to balance safety and energy consumption in rural scenarios in light and heavy traffic scenarios. Safety is quantified by the packet receive ratio ($\mathrm{PRR}$) against the minimum communication distance $D_{\mathrm{comm}}$, defined as the distance that the vehicle travels during the transmission of the safety message. Results show that, under heavy traffic, increasing $P_{\mathrm{t}}$ and selecting a low-rate $\mathrm{MCS}$ at $\mathrm{SCS} = 30$ kHz sustains high $\mathrm{PRR}$ at $D_{\mathrm{comm}}$, albeit with higher energy cost. In light traffic, maintaining lower $P_\mathrm{t}$ with low $\mathrm{MCS}$ levels achieves a favorable reliability-energy trade-off while preserving acceptable $\mathrm{PRR}$ at $D_{\mathrm{comm}}$. These findings demonstrate the necessity of adaptive, energy-aware strategy to guarantee both safety and energy efficiency in rural V2X systems.
    \end{abstract}
	

    \maketitle
	
	\section{Introduction}
	Connected and autonomous vehicle (CAV) technologies are advancing rapidly in urban areas, whilst rural deployment continues to lag behind. One of the primary barriers is the limited interest from network operators and policy makers in investing in dedicated infrastructure for sparsely populated regions, where naturally, the return on investment is often less attractive. However, rural road users expect the same level of connectivity safety as those in urban settings \cite{Scenario}. This growing expectation raises an important research question about how to ensure a reliable and robust level of CAV safety in rural areas, where communication resources are limited. Rather than relying on the expansion of infrastructure from urban to rural, it becomes critical to explore how to maximize the efficiency of available rural resources to support Vehicle-to-Everything (V2X)-based safety systems in rural driving scenarios. In this context, it is vital to address the gap by exploring strategies for a resource-efficient V2X safety deployment, which is specifically tailored to the unique challenges of rural environments.
	
	Typically, CAVs rely on integrated sensors such as cameras, radar, and lidar to process braking decisions by inferring data from their surroundings. However, in complex scenarios such as roundabouts, road merges, and junctions, sensor accuracy can degrade due to obstructions or observation failures, exposing a key limitation of sensor-based systems \cite{Sensorfusion}. Studies show that safe braking distances in such environments are greatly improved through connectivity, which allows vehicles and drivers to be alerted to potential hidden hazards in advance \cite{Jason}. Cooperative V2X communication, including Vehicle-to-Vehicle (V2V) and Vehicle-to-Infrastructure (V2I), extends environmental awareness and enables earlier, safer braking decisions \cite{Humanjiangbrake}. Among the available V2X technologies, DSRC-based V2V and 5G NR V2X are the most widely adopted. Compared to DSRC V2V, 5G New-Radio (NR) V2X provides enhanced performance, offering higher data rates, lower communication latency, greater scalability, and more robust support for high-mobility traffic scenarios \cite{5GVSDSRCHOLOGRAMS}. These advantages enable 5G NR V2X to better meet the stringent safety requirements of CAV systems \cite{ToDSRCor5G}. Across all 5G NR V2X technologies, Sidelink V2X offers a more reliable alternative by supporting direct, low-latency vehicle communication without relying on fixed infrastructure \cite{sidelinkv2v}. This makes Sidelink V2X a robust solution for maintaining connectivity and enhancing reaction times, even in sparse or infrastructure-limited environments \cite{sidelinkv2x}.

    While V2X communication is primarily designed to improve traffic efficiency and road safety, the energy cost associated with V2X infrastructure has been considered from multiple perspectives, including sensing and perception workloads, onboard computing, and system-level power management strategies \cite{R2C1_1,R2C1_2}. Among these components, the radio frequency (RF) communication subsystem remains an important and controllable contributor to energy consumption, as transmission power, and retransmissions are directly influenced by communication configurations, channel and traffic conditions \cite{R2C1_3,R2C1_4}. The energy cost associated with maintaining reliable wireless links becomes especially critical in rural deployments. In such environments, road side infrastructure is often sparse and may rely on limited power sources. Unlike sensing or control operations that are event-driven, wireless communication requires continuous use of energy for data transmission and reception. Therefore, reducing the RF energy is essential to ensure long-term system sustainability and cost-effectiveness in rural areas. Communication control factors, including subcarrier spacing ($\mathrm{SCS}$), transmit power, and modulation and coding scheme ($\mathrm{MCS}$) play a direct role in balancing RF energy consumption and communication reliability. Wider $\mathrm{SCS}$ and higher $\mathrm{MCS}$ values enable higher data throughput and can reduce transmission time, which lowers RF energy consumption per bit when the channel is stable~\cite{SCSMCS}. However, these configurations are less robust in the presence of interference or fading, leading to potential retransmissions and energy waste, especially when link quality is poor. Increasing transmit power improves signal robustness, but adds to the overall energy cost, which is particularly undesirable in power-constrained rural infrastructure. When it comes to traffic conditions, they are typically classified as “\textit{light}” or “\textit{heavy}”. Light traffic, common on rural roads, often features low vehicle density and higher average speeds, where wider $\mathrm{SCS}$ and higher $\mathrm{MCS}$ can be exploited for efficient communication. In contrast, temporary congestion or mixed traffic conditions may demand more conservative configurations to preserve reliability, especially under variable channel conditions. Given the unique challenges of rural V2X communication, such as limited infrastructure, dynamic vehicle movement, and limited energy resources, it is essential to systematically analyse how communication control settings affect both energy efficiency and safety-critical performance. Such analysis supports the design of adaptive, energy-aware communication strategies that ensure reliable operation without overprovisioning power or compromising safety.
	
	In this paper, building on our previous work~\cite{ZhanleIVs}, we extend the investigation to more complex rural traffic scenarios. First, we analyse minimum safety-critical distances that the ego vehicle travels while the braking signal is being transmitted. This analysis is conducted under varying rural traffic conditions to characterise the system’s behaviour in dynamic environments. Second, by evaluating the packet reception rate ($\mathrm{PRR}$) with respect to minimum safety-critical distances, we analyse the impact of $\mathrm{SCS}$, $\mathrm{MCS}$, and \texttt{RSU} transmit power ($P_{\mathrm{t}}$) on safety-relevant reliability. This analysis performs design-space exploration over predefined 5G NR communication control configurations to identify parameter combinations that support reliable communication and maintain a safe critical distance under diverse rural V2X traffic conditions. Finally, we conduct an energy consumption analysis under both light and heavy traffic scenarios to determine energy-efficient communication strategies that maintain safety-critical performance while minimising energy cost, and provide guidance on configuration selection based on observed safety performance to energy trade-offs.

	\subsection{Related Research}
	Connected braking is considered as a safety feature used for assisting a driver, a remote operator, or an autonomous controller to stop the vehicle in advance through V2X communications. From the communication perspective, latency is a key safety metric, as it determines how quickly safety-relevant information is delivered. In practice, packet reliability and retransmissions further affect whether delivery occurs within the required time window. From the vehicle perspective, speed is particularly critical in rural environments, as it directly determines braking distance and the required safety margin. In our previous work \cite{ZhanleIVs}, we developed a metric, so-called \( D_{\mathrm{comm}} \), as the minimum safe critical distance that the ego vehicle traveled during the time that braking signals are transmitted. It considers a scenario that the ego vehicle, due to its sensors, is unaware of the sudden braking action of a vehicle ahead in a rural setting. As a result, a crash is imminent unless an avoidance maneuver or braking action is executed. The metric \(D_{\mathrm{comm}}\) implicitly incorporates these factors by representing the maximum distance over which safety-relevant information can be exchanged with sufficient reliability, thereby mapping communication performance into a spatial safety margin that is directly interpretable in terms of vehicle dynamics and safety requirements. In the context of connected braking, the ego vehicle communicates with the target vehicle or the \texttt{RSU} and triggers the brake based on delivered data packets. A larger packet delay leads to a slower connected braking, which increases \( D_{\mathrm{comm}} \) accordingly. In general, the increase of \( D_{\mathrm{comm}}\) is mainly caused by the increase of packet error rate ($\mathrm{PER}$), which is defined as the total number of received error packets against the total number of packet transmitted, and it is the main indicator of V2X propagation delay \cite{PERinWCL}. In \cite{ZhanleIVs}, we investigate the impact of transmit power on the V2V $\mathrm{PER}$ in a DSRC channel. An analysis on \( D_{\mathrm{comm}} \) is provided to demonstrate connected braking requirements within a $100$ \si{\meter} in a rural road-following scenario.
 
    The performance evaluation of communication control factors in V2X communication remains an open research challenge \cite{v2xchallenge}. In practice, the increasing complexity of road environments, driven by a higher number of road users, additional lanes, extended road lengths, and higher vehicle speeds make it even more difficult to maintain the required levels of safety and reliability. This challenge is closely linked to the allocation of communication control factors, such as transmit power, $\mathrm{SCS}$, and $\mathrm{MCS}$, which must be dynamically adjusted to match traffic conditions while avoiding excessive energy consumption. As the number of vehicles participating in the V2X network increases, the communication system must handle significantly higher data volumes, which can lead to network congestion, packet loss, and increased message delivery delays~\cite{DSRCplatooning}. These conditions directly impact the reliability of V2X connections, especially in environments with a lack of communication resource such as rural. In parallel, higher initial vehicle speeds further complicate communication demands. The authors of \cite{V2Xbraking} show that increased speeds require faster reaction times, which in turn demand a higher signal-to-noise ratio (SNR) to maintain low latency performance and increase the required \(D_{\mathrm{comm}}\). This adds pressure on the critical-safety system to deliver Basic Safety Messages (BSMs) with minimal delay~\cite{BSMs}. Collectively, these factors increase the likelihood of missed or delayed hazard alerts in complex, high-mobility scenarios, making the real-time monitoring and control of \(D_{\mathrm{comm}}\) more challenging. These challenges highlight the need for a proper setting of communication factors that can ensure timely and reliable message delivery under varying traffic and mobility conditions.

The impact of communication control factors on V2X communication has been investigated in a number of studies \cite{V2XRA-Pt,V2XRA-SCS,V2XRA3,V2Xsurvey,V2XRA-MCS-NR1,V2XRA-MCS-NR2}. The authors of \cite{V2XRA-Pt} have investigated the impact of \( {P}_{\mathrm{t}} \) on communication quality in urban scenario. It shows that a higher \( {P}_{\mathrm{t}} \) improves the received SNR, especially over long distances or in high interference environments. The research proposed in \cite{V2XRA-SCS} specifically evaluated the impact of $\mathrm{SCS}$ on V2X communication performance. Using a relatively high $\mathrm{SCS}$, such as $30$ \si{\kilo\hertz}, shortens the duration of each symbol, making it particularly suitable for low-latency applications. The study introduced in \cite{V2XRA3,V2Xsurvey} evaluates V2X communication performance under different $\mathrm{SCS}$ and $\mathrm{MCS}$ configurations in urban scenarios. Their findings indicate that a proper tuning of $\mathrm{SCS}$ and $\mathrm{MCS}$ enables the system to better adapt to varying traffic loads and channel conditions, thereby supporting more responsive and reliable data exchange, particularly under high vehicular densities. The research~\cite{V2XRA-MCS-NR1,V2XRA-MCS-NR2} focus on $\mathrm{MCS}$ within urban 5G NR V2X scenarios. It shows that appropriately adapting the $\mathrm{MCS}$ to match prevailing channel conditions can significantly reduce PER. Despite providing useful insights, the above works \cite{V2XRA-Pt,V2XRA-SCS,V2XRA3,V2Xsurvey,V2XRA-MCS-NR1,V2XRA-MCS-NR2} do not address a joint performance evaluation of \( {P}_{\mathrm{t}} \), $\mathrm{SCS}$, and $\mathrm{MCS}$. This gap is especially notable in the context of 5G NR V2X for rural traffic scenarios, where the unique channel characteristics demand a holistic approach to communication control. Another notable gap in \cite{V2XRA-Pt,V2XRA-SCS,V2XRA3,V2Xsurvey,V2XRA-MCS-NR1,V2XRA-MCS-NR2} is the lack of consideration of RF energy consumption in the performance evaluation of communication control factors. Although previous works have focused on improving reliability and latency, it has largely overlooked the trade-offs between performance and energy efficiency, which is an increasingly important concern, particularly in resource-constrained V2X deployments such as rural traffic scenarios.

Several studies have considered the energy aspect of V2X communication, recognising its importance in sustaining large-scale, real-time vehicle connectivity \cite{energy1,V2XRA-Pt2,EnergySCS,EnergySCS2,EnergyMCS}. These works typically explore how different communication control factors can be evaluated to reduce RF energy consumption without compromising reliability. The energy-efficient resource allocation problem in urban V2X scenarios is considered in \cite{energy1,V2XRA-Pt2}. It indicates that the appropriate control of \( {P}_{\mathrm{t}} \) and interference management are important for energy-efficient V2X. In addition, $\mathrm{SCS}$ and $\mathrm{MCS}$ can also come at the cost of increased power consumption. The research conducted in \cite{EnergySCS} shows that higher $\mathrm{SCS}$, such as $30$ \si{\kilo\hertz}, can lead to higher energy requirements for accurate transmission and reception due to increased signal bandwidth and processing demands, especially under different traffic conditions. In contrast, it was shown that such low-priority and delay-tolerant applications as traffic monitoring or non-time-sensitive updates, benefit from lower $\mathrm{SCS}$ such as $15$ \si{\kilo\hertz} \cite{EnergySCS2}. Moreover, lower $\mathrm{MCS}$ levels can conserve energy and minimise resource usage. As demonstrated in \cite{EnergyMCS}, selecting a lower $\mathrm{MCS}$ in challenging environments leads to improved transmission efficiency and reduced RF energy consumption without significantly compromising throughput. However, these studies \cite{energy1,V2XRA-Pt2,EnergySCS,EnergySCS2,EnergyMCS} often overlook rural environments, where vehicles are typically more spaced out and signals travel over longer distances with fewer obstructions. Such conditions demand a joint performance evaluation of all three communication control factors to ensure reliable communication while maintaining energy efficiency.

Motivated by these research gaps, this paper investigates the impact of three communication control factors \( {P}_{\mathrm{t}} \), $\mathrm{SCS}$, and $\mathrm{MCS}$ on RF energy consumption and communication performance in different rural V2X scenarios. Performance is evaluated using $\mathrm{PRR}$ and the \( D_{\mathrm{comm}} \), reflecting the ability to deliver safety messages within safety-critical distances. By analysing $\mathrm{PRR}$ and RF energy consumption under different traffic densities and vehicle speeds, the study aims to provide configuration guidance that ensure reliable, low-latency communication while minimising energy expenditure.

    \begin{figure}[t]
		\centerline{\includegraphics[width=0.5\textwidth]{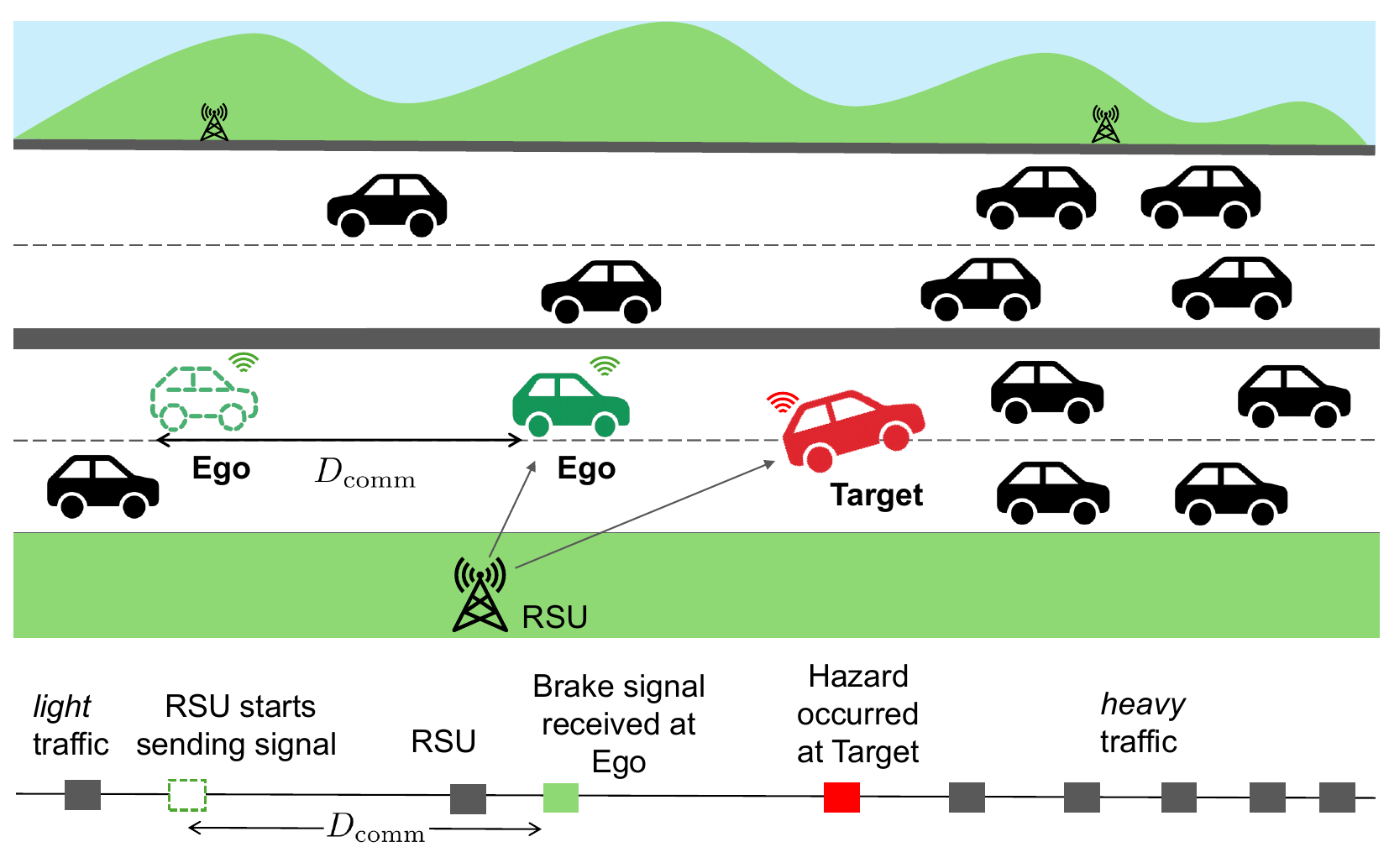}}
		\centering
		\caption{Traffic system model where $\texttt{RSU}$ sends a braking signal to $\texttt{Ego}$ to activate the brake. $D_\mathrm{comm}$ is shown as the distance that $\texttt{Ego}$ travels during the transmission.}
		\label{scenario}
	\end{figure}

	\begin{figure}[t]
        \centerline{\includegraphics[width=0.5\textwidth]{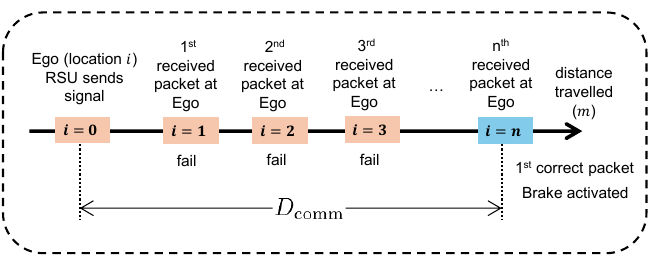}}
		\centering
		\caption{The process of sending and receiving safety signals between $\texttt{RSU}$ and $\texttt{Ego}$.}
		\label{Dcommstamp}
	\end{figure}
	\subsection{Main Contributions}
	The main contributions of this paper now can be summarised as follows. \textbf{Firstly}, we investigate the impact of three key communication control factors including $\mathrm{SCS}$, \( {P}_{\mathrm{t}} \) of road side unit, and $\mathrm{MCS}$ on safety and RF energy consumption in a connected braking scenario. Our rural scenario represents random layouts in \textit{light} and \textit{heavy} traffic conditions, where CAVs exchange time-critical braking messages via 5G NR Sidelink V2X communication. Our \textbf{second} contribution is the development of analytical formulations for the \(D_{\mathrm{comm}}\), $\mathrm{PRR}$ and RF energy consumption, which together describe the achievable braking reliability over distance. These formulations provide a comprehensive analysis, linking communication control factors to safety-critical performance indicators. \textbf{Finally}, the impact of communication control factors on the safety and RF energy consumption is validated through extensive simulations, offering valuable insights into the trade-offs between safety and RF energy consumption. Using the analytical results, we identify the appropriate configurations of $\mathrm{SCS}$, \( {P}_{\mathrm{t}} \), and $\mathrm{MCS}$ that contributes to an improved safety and energy performance. The analysis forms a baseline for adaptive control strategies on them in future rural CAV safety applications.

    The remainder of this paper is organised as follows. Section~\ref{se:system model} introduces the system model and the simulation setup. In Section~\ref{se:Performance Formulation}, we formulate \(D_{\mathrm{comm}}\) and RF energy consumption. In Section~\ref{se:simulation results}, we illustrate the simulation results and explore the impact on $\mathrm{PRR}$, \(D_{\mathrm{comm}}\) from tuning communication control factors. Section~\ref{se:analysis} provides quantitative analysis on the results, and impact of communication control factors on safety and RF energy consumption. Finally, Section~\ref{se:conclusion} summarises the main findings of this paper.

	\section{System Model}
	\label{se:system model}
	\subsection{Traffic Model and Critical Distance}
    \label{se:system model A}
	We consider a traffic consisting of a road side unit ($\texttt{RSU}$)\footnote{In our rural scenario, \texttt{RSU} is assumed to operate as a communication-centric V2X infrastructure node, focusing on sidelink communication support rather than environmental perception or sensor-based detection.}, an ego vehicle ($\texttt{Ego}$) and a target vehicle ($\texttt{Target}$) in a rural area. All the vehicles are registered with $\texttt{RSU}$, as illustrated in Fig.~\ref{scenario}. Considering a scenario where $\texttt{Target}$ is undergoing a hazard, as a result, $\texttt{RSU}$ notified $\texttt{Ego}$ by sending a brake signal via the communication channel. In this context, $\texttt{RSU}$ transmits data packets containing brake signals, and $\texttt{Ego}$ needs to receive at least one correct data packet to perform the brake. Error packets received by $\texttt{Ego}$ causes a delay due to retransmission. Since the time $\texttt{RSU}$ sent out the brake signals, $\texttt{Ego}$ travels a distance until it first received a correct signal. This distance is known as a critical distance, e.g., \( D_{\mathrm{comm}} \). To explain, $D_\mathrm{comm}$ demonstrates the distance that $\texttt{Ego}$ travels during the time required by signal retransmission due to error, as illustrated in Fig.~\ref{Dcommstamp}. At location \(i=0\), $\texttt{RSU}$ transmits a safety message in terms of packet to $\texttt{Ego}$ periodically. $\texttt{Ego}$ receives these packets sequentially at location indices \(i=1,2,3,\ldots\). Due to channel errors, the process continues until the first error-free packet is successfully received at index \(i=n\), which then provides the valid brake signal. During this interval, $\texttt{Ego}$ continues to travel at its current speed, covering a distance referred to as \(D_{\mathrm{comm}}\). In this context, \( D_{\mathrm{comm}} \) can be considered as a safety threshold distance. The connected braking can execute successfully only when the distance between $\texttt{RSU}$ and $\texttt{Ego}$ is larger than \( D_{\mathrm{comm}} \).

	\subsection{Simulation Setup}
		\begin{table}[t]
		\centering
		\begin{minipage}{0.475\textwidth}
			\centering
			\caption{SIMULATION PARAMETERS}
			\label{tab:table1}
			\renewcommand{\arraystretch}{0.75}
			\begin{tabular}{@{}p{4.0cm}p{3.9cm}@{}}
				\hline
				\hline
				\textbf{Traffic Scenario} & \\
				Scenario & $2$ lanes rural highway, $2000$ \si{m} \\
				Vehicle density ($\rho$) & $30$ - $100$ vehicles per~km\\
				Average speed ($v$) & $50$ - $110$ \si{\kilo\metre/\hour} \\
				Mobility pattern ($V$) &$V \sim \mathcal{N}(v,\, 7^2)\ \text{km/h}$ \cite{WiLab2} \\
				Packet generation & Every $100$ \si{\milli\second} \\
				
				\hline
				\textbf{Communication Settings} & \\
				Center frequency $(f_c)$ & ITS bands at $5.9$ \si{\giga\hertz} \\
				Bandwidth & $20$ \si{\mega\hertz} \\
				Transmission power density & $13$ \si{\dBm}/\si{\mega\hertz} \\
				Antenna gain & $3$ \si{\dBi} \\
				Noise figure & $9$ \si{\dB} \\
				Propagation model & 3GPP Rural \cite{3GPP} \\
				Shadowing & Variance $3$ \si{\dB}, decorrelation distance $25$ \si{\meter} \cite{WiLab1} \\
				Packet size & $350$ bytes \\
				Number of packets per second & $10$\\
				Hybrid Automatic Repeat Request  & Blind retransmissions, max. $2$ \\
                Average circuit power per transmission $(P_c)$& $100$ \si{\mW} \cite{R2C1_2}\\
                Average idle power $(P_i)$ & $125$ \si{\mW} \cite{R2C1_1}\\
				\hline
				\textbf{Communication Control Factors} & \\
				$\mathrm{MCS}$ & $8$-$10$ (QPSK), $12$-$18$ (16-QAM) \\
				$\mathrm{SCS}$ & $15$ and $30$ \si{\kilo\hertz}\\
				\( {P}_{\mathrm{t}} \) & $23$, $24$, $25$, $26$ \si{\dBm}\\
				Subchannel size & $10$ Physical Resource Blocks (PRB) \\
				Number of subchannels & $5$ \\
				Subchannels per packet & $2$ \\
				Resource keep probability & $0.4$ \\
				Number of Demodulation Reference Symbols (DMRS) per slot& $24$ for $\mathrm{SCS} = 15$, \ \ $18$ for $\mathrm{SCS} = 30$\\
				Allocation periodicity & $100$ \si{\milli\second} \\
				Sensing threshold & $–110$ \si{\dBm} \\
				\hline
				\hline
			\end{tabular}
		\end{minipage}
	\end{table}

	The V2X system in our rural scenario, considering traffic complexity and communication reliability, is simulated using WiLabV2XSim \cite{WiLab1}, an open source simulator that is well known for its accurate modelling of vehicular mobility, wireless channel dynamics, and communication reliability under varying traffic conditions. It provides a flexible and extensible framework for modelling V2X networks under various communication protocols, channel conditions, and mobility patterns~\cite{WiLab2, WiLab3}.
    
    In this study, to approximate rural traffic conditions, our rural highway scenario has a total road length of $2000$ \si{\meter} and two lanes ($2 + 2$) both ways. Each lane is $4$ \si{\meter} wide and the road speed limit is set to $120$ \si{\kilo\metre/\hour}. We consider light traffic at $30$ and $50$, and heavy traffic at $80$ and $100$ vehicles per \si{\kilo\meter}. Vehicles moving at an average speed of $50$ to $110$ \si{\kilo\metre/\hour}, which follows a Gaussian
	distribution with a standard
	deviation of $7$ \si{\kilo\metre/\hour}. 
	
	The channel model setting is generated following the 3GPP rural \cite{3GPP}. Various $\mathrm{MCS}$ values are considered in the simulation: $\mathrm{MCS}$ $8$ to $10$ represents the Quadrature Phase Shift Keying Modulation (QPSK), and $\mathrm{MCS}$ $12$ to $18$ perform the $16$-Quadrature Amplitude Modulation ($16$ QAM). In addition, $\mathrm{SCS}$ is modified between $15$ and $30$ \si{\kilo\hertz} and \( {P}_{\mathrm{t}} \) is modified from $23$ to $26$ \si{\dBm}. During the sensing window, \texttt{RSU} measures the received signal reference power (RSRP) associated with decoded sidelink control information (SCI), and the sensing threshold  is set to $-110$~\si{dBm} in this study. Different vehicle densities cause varying channel congestion, impacting delay and \(\mathrm{PRR}\), and thereby altering the performance evaluation of communication control factors. In this work, we also consider shadowing. The value of shadowing \( S_i \) is updated to its new value \( S_{i+1} \) following the distance of the vehicle traveling as follows \cite{3gpp36885}:
	\begin{equation}
	S_{i+1} = \exp \Big(-\frac{d}{d_{\mathrm{corr}}}\Big) S_i + \sqrt{1 - \exp \Big(-\frac{2d}{d_{\mathrm{corr}}} \Big)} \ N_{i+1},
	\label{eq1}
	\end{equation}
	where $d$ [\si{meters}] is the change in distance between the \texttt{RSU} and the \texttt{Ego}. \( d_{\mathrm{corr}} \) denotes the decorrelation distance of $25$ \si{\meter}. Beyond \( d_{\mathrm{corr}} \) the shadowing becomes increasingly uncorrelated and varies more independently. In addition, \( N_{i+1} \) stands for a log normal independent random variable with standard deviation \( \sigma =3\). We also use the free space pathloss (PL) model specified in 3GPP~\cite{3GPP}:
	
	\begin{equation}
	\text{PL}[\si{dB}] = 20 \log_{10}(d) + 20.0 \log_{10}(f_c)+ 32.45,
	\label{eq2}
	\end{equation}
	where $f_{c}$ is the center frequency in \si{\giga\hertz}. The simulation parameters for the traffic scenario, communication settings and communication control factors are summarised in Table~\ref{tab:table1}. Specific parameters will be provided in each simulation.

    \section{Performance Formulation}
    \label{se:Performance Formulation}
    This section aims to formulate the received ${E_b}/{N_0}$, the data rate, the critical distance $D_\mathrm{comm}$ and the total energy consumption.
    \subsection{The Analysis of ${E_b}/{N_0}$ and Data Rate}
    The maximum Bit Error Rate ($\mathrm{BER}$) against the V2X distance can be computed using the following equation:
	\begin{equation}
	\mathrm{BER}(d_i) = 1 - \Big[1-\mathrm{PER}(d_i) \Big]^{\frac{1}{L}} = 1 - \mathrm{PRR}(d_i)^{\frac{1}{L}},
	\end{equation}
	where \( L \) is the packet length. The $\mathrm{BER}$ can also be formulated as
	
	\begin{equation}
	\mathrm{BER}(d_i, M) =
	\begin{cases} 
	Q\left(\sqrt{\frac{E_b}{N_0}(d_i)}\right), &  M = 4  \\[8pt]
	\frac{3}{8} Q\left(\sqrt{\frac{4}{5} \cdot \frac{E_b}{N_0}(d_i)}\right), &  M = 16,
	\end{cases}
	\end{equation}
	where $Q(.)$ denotes the Q-function. Therefore, to determine the required \(E_b/N_0\), we use the inverse Q-function, as follows
	\begin{equation}
	\!\!\!\!\frac{E_b}{N_0}(d_i, M) =
	\begin{cases} 
	\left[ Q^{-1} \big( \mathrm{BER}(d_i, M) \big) \right]^2, &  \!\! M = 4 \\[8pt]
	\frac{5}{4} \left[ Q^{-1} \big(\frac{8}{3}\ \mathrm{BER}(d_i, M) \big) \right]^2, &  \!\! M = 16.
	\end{cases}
	\end{equation}
	where \( Q^{-1} \) represents the inverse Q-function.
    The data rate can be expressed as:
	\begin{equation}
	R =  12 \ N_{\mathrm{sub}} \ N_{\mathrm{bits/symbol}} \ R_c \ \text{$\mathrm{SCS}$} \ N_{\mathrm{symbols}},
	\end{equation}
	where $N_{\mathrm{sub}}$ denotes the number of allocated PRBs, $N_{\mathrm{bits/symbol}}$ represents the number of bits per symbol dictated by the modulation scheme while $R_c$ is the coding rate and $N_{\mathrm{symbols}}$ is the number of symbols per slot.

    \subsection{The Analysis of Critical Distance}
    \begin{lemma}[\textbf{The Formulation of Critical Distance}]
        Let $v$ stands for the average speed of \texttt{Ego} and $N$ be the total number of delivered packets.  Mathematically, $D_\mathrm{comm}$ can be expressed as 
        \begin{equation}
            {D}_{\mathrm{comm}} = \frac{v}{\mathrm{pps}}+\frac{(N-1)v}{\mathrm{pps}} \big(1-\mathrm{PRR} \big).
        \label{eqDcomm}
        \end{equation}
    \end{lemma}
    \begin{proof}
        Packets are transmitted periodically at rate $\mathrm{pps}$ packets per second, so each transmission slot lasts $1/\mathrm{pps}$ seconds. In this work, for tractable analysis, we assume that each packet reception is an independent and identically distributed (i.i.d) Bernoulli trial with the same success probability \cite{iid1}, the \texttt{Ego} vehicle is assumed to travel a set distance during the packet transmission attempt of the first packet, the expected number of failed receptions over the next $N-1$ attempts is $(N-1)(1-\mathrm{PRR})$. Each failure delays the first correctly received packet by exactly one slot; hence the expected additional waiting time (relative to the first attempt) is 
        \begin{align}
            \mathbb{E}[T_{\mathrm{extra}}]=\frac{1+(N-1)(1-\mathrm{PRR})}{\mathrm{pps}}~.
        \end{align}
        If the \texttt{Ego} vehicle travels at approximately constant speed $v$ during this interval, the expected additional distance covered before a correct packet is received is
        \begin{align}
            D_{\mathrm{comm}} \;=\; v\,\mathbb{E}[T_{\mathrm{extra}}]
    \;=\; \frac{v}{\mathrm{pps}}+\frac{(N-1)v}{\mathrm{pps}} \big(1-\mathrm{PRR} \big).
        \end{align}
        The proof is now completed.
    \end{proof}

    \subsection{The Analysis of Total Energy}
    \begin{lemma}[\textbf{Total energy spent on RF transmissions with truncated HARQ}]
    Let $\mathrm{pps}$ be the packet generation rate [pkt/s], $T$ be the horizon length [s]. Hence, there are $N_{\mathrm{pkt}}=\mathrm{pps}\,T$ packets generated. In addition, $L_{\mathrm{bits}}$ is the payload length [bits], $R$ denotes data rate [bit/s] and $H\in\mathbb{N}$ is the maximum number of attempts per generated packet (blind HARQ cap). Then, the expected total RF transmit energy over the horizon is
    \begin{align} \label{eq:etotal}
        E_{\mathrm{total}} = N_{\mathrm{pkt}}\, P_t\,\frac{L_{\mathrm{bits}}}{R}\;
   \frac{1-(1-\mathrm{PRR})^{H}}{\mathrm{PRR}}.
    \end{align}
    \end{lemma}

    \begin{proof}
    Let $K\in \mathbb{N}^+$ denote the number of transmission attempts until the first correct reception for a given generated packet. Under the standard i.i.d.\ success model with success probability $\mathrm{PRR}$ per attempt, $K$ is geometrically distributed with the following Probability Mass Function (PMF):
    \begin{align}
        \mathbb{P}(K=k)=(1-\mathrm{PRR})^{k-1}\,\mathrm{PRR},\qquad k=1,2,\ldots
    \end{align}
    and thus $\mathbb{P}(K\ge h)=(1-\mathrm{PRR})^{h-1}$ for $h\ge1$. Because at most $H$ attempts are undertaken per packet, the number of actual transmissions devoted to that packet is $\min(K,H)$. By the tail-sum identity for nonnegative integer-valued random variables, we have
    \begin{equation}
        \begin{aligned}
        \!\!\!\!\mathbb{E}[\min(K,H)] \!& = \sum_{h=1}^{H}\mathbb{P}\{K\ge h\} \\
        &= \sum_{h=1}^{H}(1-\mathrm{PRR})^{h-1} \frac{1-(1-\mathrm{PRR})^{H}}{\mathrm{PRR}}.
        \end{aligned}
    \end{equation}
    Note that each attempt lasts $t_{\mathrm{tx}}=L_{\mathrm{bits}}/R$ seconds and consumes $E_{\mathrm{att}}=P_t\,L_{\mathrm{bits}}/R$ Joules. Hence, the expected transmit energy per generated packet is
    \begin{align}
    E_{\mathrm{pkt}}
    = P_t\,\frac{L_{\mathrm{bits}}}{R}\,
  \frac{1-(1-\mathrm{PRR})^{H}}{\mathrm{PRR}}.
    \end{align}
    Since there are $N_{\mathrm{pkt}}$ generated packets, $E_\mathrm{total} = N_\mathrm{pkt} E_\mathrm{pkt}$, which proves the claim.
    \end{proof}

    \begin{remark}
    In the special case when only a single attempt is allowed (e.g., \(H=1\)), the expected total transmit energy reduces to \(N_\mathrm{pkt} P_t\,L_{\mathrm{bits}}/R\), which is independent of \(\mathrm{PRR}\). 
    In the case when unlimited retransmissions are allowed (e.g., \(H\to\infty\)), $(1 - \mathrm{PRR})^H \to 0$. Therefore, the expected energy per generated packet becomes \(N_\mathrm{pkt}P_t\,L_{\mathrm{bits}}/R/\mathrm{PRR}\).
    \end{remark}

\begin{lemma}[\textbf{Total communication energy}]
\label{lem:etotal_comm}
let $P_{\mathrm{c}}$ denote the additional averaged circuit power consumed during transmission.
Let $P_{\mathrm{i}}$ denote the averaged platform idle power over the horizon $T$.
Then the expected total communication energy over the horizon is
\begin{align}
\label{eq:etotal_comm}
\!\!\!E_{\mathrm{comm}}
&= N_{\mathrm{pkt}}\!\left[(P_t+P_c)\frac{L_{\mathrm{bits}}}{R}
\frac{1-(1-\mathrm{PRR})^{H}}{\mathrm{PRR}}
+ \frac{P_{\mathrm{i}}}{\mathrm{pps}}\right]. \nonumber
\end{align}
\end{lemma}

\begin{proof}
Based on Lemma 2, during each attempt, an averaged circuit power $P_c$ is associated with the active transmission. Therefore, the expected communication energy spent on transmissions for one generated packet is
\begin{equation}
\begin{aligned}
\label{eq:circuitpower}
E_{\mathrm{pkt,tx}}= (P_t+P_c)\,\frac{L_{\mathrm{bits}}}{R}\,\frac{1-(1-\mathrm{PRR})^{H}}{\mathrm{PRR}},
\end{aligned}
\end{equation}
Next, the corresponding idle energy over the horizon is $E_{\mathrm{idle}}=P_i\,T$, and in per-packet allocation:
\begin{align}
\label{eq:idlepower}
E_{\mathrm{idle}}
= P_i\,T
= N_{\mathrm{pkt}}\;\frac{P_i}{\mathrm{pps}},
\end{align}
Combing ~\eqref{eq:circuitpower} and ~\eqref{eq:idlepower} gives us
\begin{align}
E_{\mathrm{comm}}
&= N_{\mathrm{pkt}}\left(E_{\mathrm{pkt,tx}}+\frac{P_i}{\mathrm{pps}}\right).
\end{align}
which claims the lemma.
\end{proof}

	\section{Simulation Results} \label{se:simulation results}
    \begin{figure}[t]
		\centerline{\includegraphics[width=0.5\textwidth]{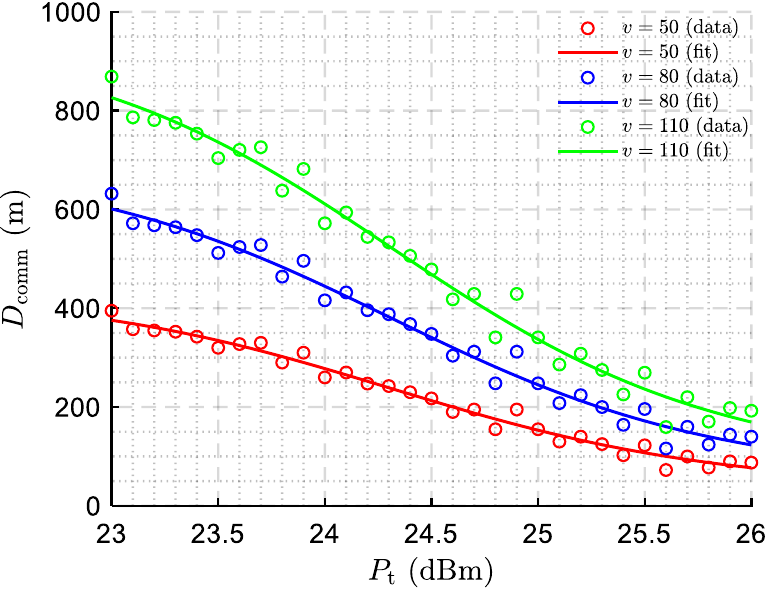}}
		\centering
		\caption{Minimum \( D_{\mathrm{comm}} \) under various values of \( {P}_{\mathrm{t}}\) and $v$.}
		\label{Dcomm}
	\end{figure}
    
	This section provides the simulation results, which are progressively structured to build on each other’s insights. We begin by analysing \( D_{\mathrm{comm}} \) to determine the safety-critical communication threshold for the upcoming simulation traffic scenarios. Subsequently, the next simulation evaluates the effectiveness of different $\mathrm{SCS}$ settings. Based on this, an analysis of $P_t$ under different traffic conditions is investigated. The next simulation focuses on analysing $\mathrm{MCS}$ effectiveness while the final simulation campaign explores the RF energy consumption analysis, incorporating the performance evaluation of communication control factors derived from the findings of the previous simulations.

    \subsection{The Analysis of Critical Distance}
    Before introducing the main simulation results, we evaluate the minimum \( D_{\mathrm{comm}} \) for various values of \({P}_{\mathrm{t}}\) ranging from $23$ to $26$ \si{\dBm}. Fig.~\ref{Dcomm} illustrates the obtained values of \( D_{\mathrm{comm}} \) with respect to \({P}_{\mathrm{t}}\) at vehicle speeds of $50$, $80$, and $110$ \si{\kilo\metre/\hour}. The results show that higher \({P}_{\mathrm{t}}\) consistently decreases \( D_{\mathrm{comm}} \) across all speed levels. This is because increasing transmit power reduces the $\mathrm{PER}$ and shortens the retransmission delays. As shown in the result, increasing \(P_{\mathrm{t}}\) from $23$ to $26$~\si{\dBm} at $50$ \si{\kilo\metre/\hour} results in a substantial reduction in \(D_{\mathrm{comm}}\) from $375$~\si{\meter} to just under $100$~\si{\meter}, which is decreased for nearly $4$ times. A similar trend is observed at higher speeds, where the required \(D_{\mathrm{comm}}\) at \(v = 110~\si{\kilo\metre/\hour}\) is approximately twice that at \(v = 50~\si{\kilo\metre/\hour}\). Across all speeds, increasing \(P_{\mathrm{t}}\) from $23$ to $26$~\si{\dBm} has effectively reduce the \(D_{\mathrm{comm}}\) for more than $100$\%. Nevertheless, the result illustrates that increasing \(P_{\mathrm{t}}\) continues to provide a significant reduction in \(D_{\mathrm{comm}}\), even under these more demanding conditions. Introducing \( D_{\mathrm{comm}} \) at this stage is essential, as it serves as a safe distance threshold for subsequent simulations. Based on this, the most appropriate communication control factor setting is defined as the configuration that achieves the highest $\mathrm{PRR}$ before the simulated distance exceeds \( D_{\mathrm{comm}} \).

	\subsection{The Impact of SCS on PRR}
    
    \begin{figure}[t]
		\centerline{\includegraphics[width=0.5\textwidth]{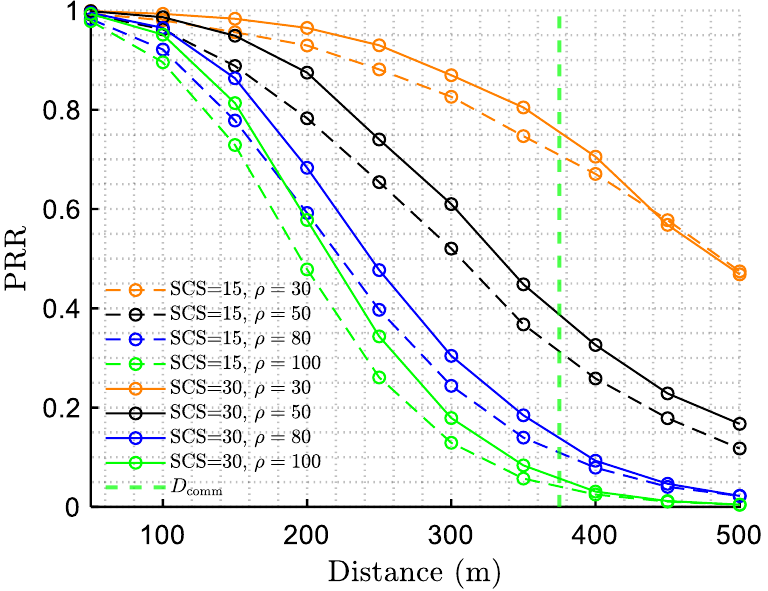}}
		\centering
		\caption{$\mathrm{PRR}$ benchmark between $\mathrm{SCS}$ = $15$ \si{\kilo\hertz} and $\mathrm{SCS}$ = $30$ \si{\kilo\hertz}. In this simulation, traffic densities are from $30$ to $100$ \si{vehicles/\kilo\metre}, with \(P_{\mathrm{t}} = 23\)~\si{\dBm} and $\mathrm{MCS}$ = $8$.}
		\label{23dBm, 30 compare 15 SCS MCS 8}
	\end{figure}
	
	\begin{figure}[t]
		\centerline{\includegraphics[width=0.5\textwidth]{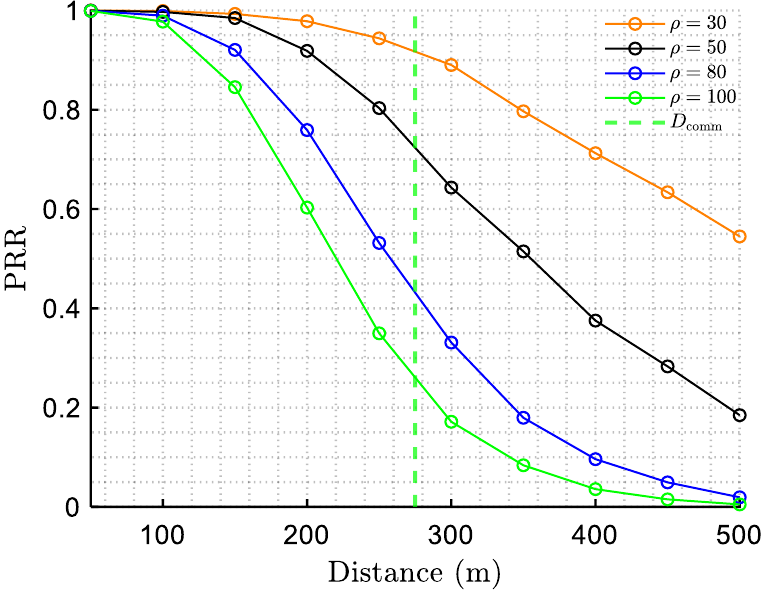}}
		\centering
		\caption{$\mathrm{PRR}$ tested with $\rho \in \{30, 50, 80, 100\}$ \si{vehicles/\kilo\metre}. In this simulation, $\mathrm{SCS}$ = $30$ \si{\kilo\hertz}, $\mathrm{MCS}$ = $8$, and \(P_{\mathrm{t}} = 24\)~\si{\dBm}.}
		\label{24 dBm}
	\end{figure}
	
	\begin{figure}[t]
		\centerline{\includegraphics[width=0.45\textwidth]{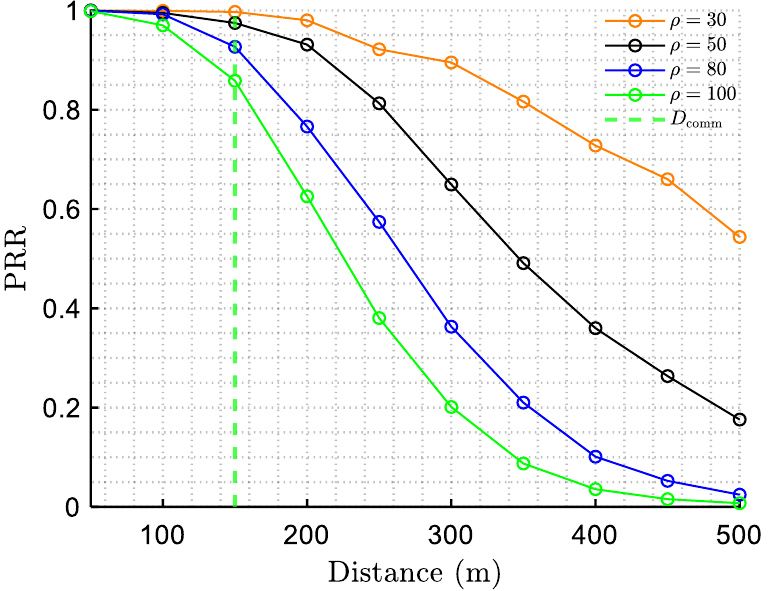}}
		\centering
		\caption{$\mathrm{PRR}$ tested with $\rho \in \{30, 50, 80, 100\}$ \si{vehicles/\kilo\metre}. In this simulation, $\mathrm{SCS}$ = $30$ \si{\kilo\hertz}, $\mathrm{MCS}$ = $8$, and \(P_{\mathrm{t}} = 25\)~\si{\dBm}.}
		\label{25 dBm}
	\end{figure}
	\begin{figure}[t]
		\centerline{\includegraphics[width=0.45\textwidth]{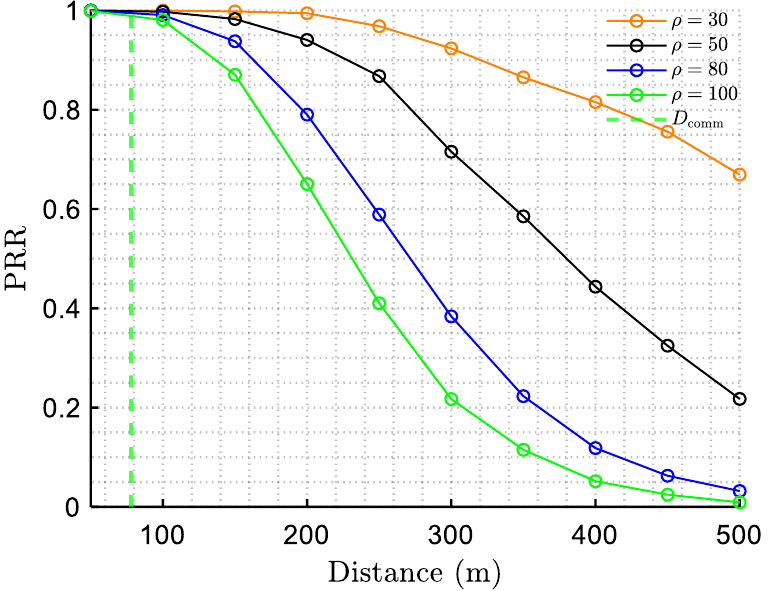}}
		\centering
		\caption{$\mathrm{PRR}$ tested with $\rho \in \{30, 50, 80, 100\}$ \si{vehicles/\kilo\metre}. In this simulation, $\mathrm{SCS}$ = $30$ \si{\kilo\hertz}, $\mathrm{MCS}$ = $8$, and \(P_{\mathrm{t}} = 26\)~\si{\dBm}.}
		\label{26 dBm}
	\end{figure}
	
	Fig.~\ref{23dBm, 30 compare 15 SCS MCS 8} shows the $\mathrm{PRR}$ performance under $\mathrm{SCS}$ values of $30$ \si{\kilo\hertz} and $15$ \si{\kilo\hertz} across different vehicle densities. $\mathrm{MCS} = 8$ is adopted in this simulation as it provides stronger error correction and higher robustness, thereby minimising external effects when evaluating $\mathrm{SCS}$ effectiveness. As can been seen, the acheived $\mathrm{PRR}$ at $\mathrm{SCS}$ of $30$ \si{\kilo\hertz} is higher than that of $15$ \si{\kilo\hertz}, due to its improved resilience to Doppler shifts and phase noise. In high-mobility and long-distance scenarios, using $30$ \si{\kilo\hertz} reduces intercarrier interference, while its shorter symbol duration mitigates multipath fading, both of which contribute to improved $\mathrm{PRR}$. Referring to Fig.~\ref{Dcomm}, the minimum \( D_{\mathrm{comm}} \) at \({P}_{\mathrm{t}}=23\) \si{\dBm} and $50$ \si{\kilo\metre/\hour} is $375$ \si{\meter}. At $375$ \si{\meter} distance, $\mathrm{SCS}$ of $30$ \si{\kilo\hertz} provides a better performance in most of the cases, compared to $15$ \si{\kilo\hertz}. Therefore, $\text{SCS} = 30$ is applied in subsequent simulations to study the impact of other communication control factors.

\begin{figure*}[t]
    \centering
    \subfloat[]{%
        \includegraphics[width=0.3\textwidth]{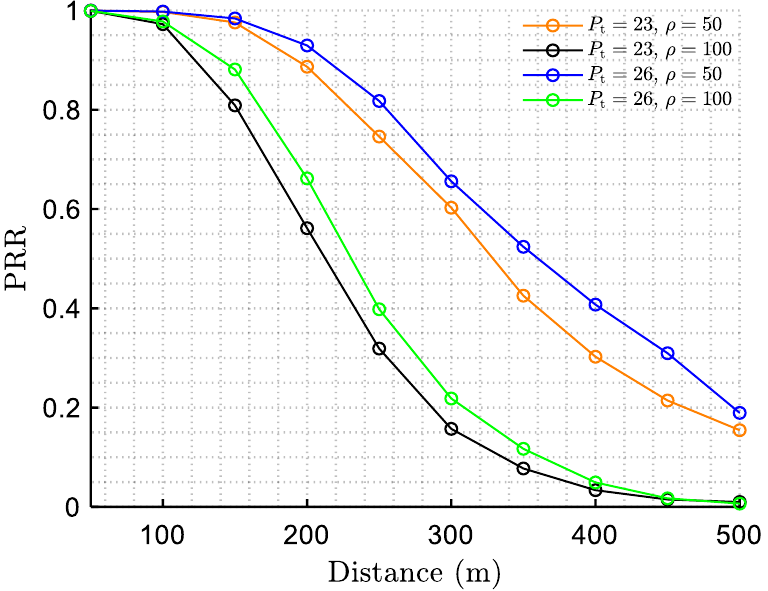}
        \label{Test PRR at 50 speed, 23 and 26 dBm}
    }
    \hfill
    \subfloat[]{%
        \includegraphics[width=0.3\textwidth]{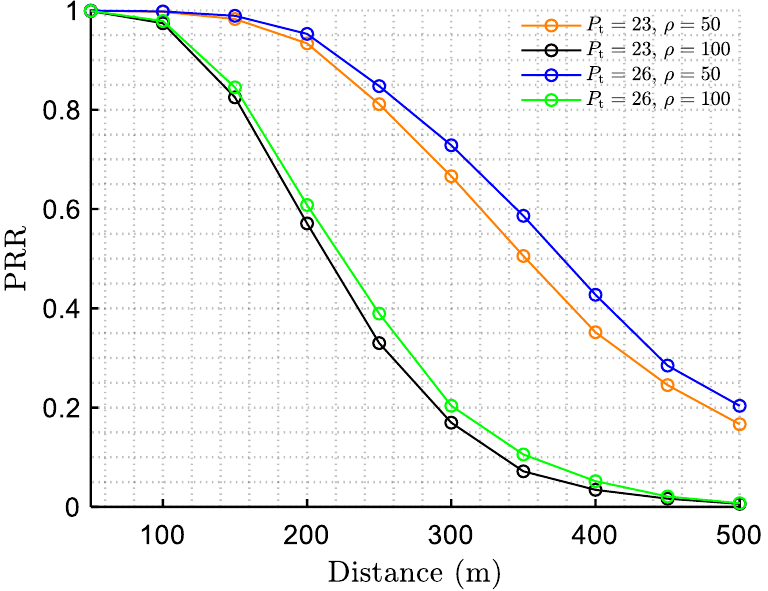}
        \label{Test PRR at 80 speed, 23 & 26 dBm}
    }
    \hfill
    \subfloat[]{%
    \includegraphics[width=0.3\textwidth]{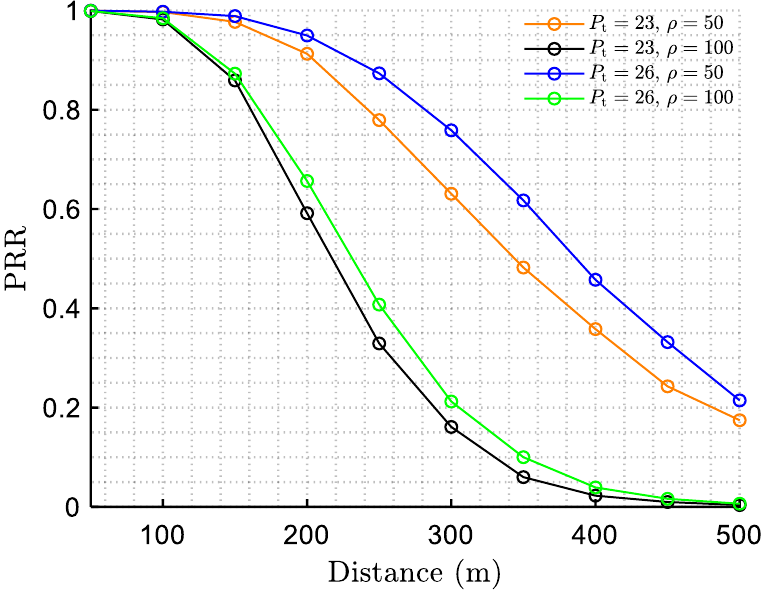}
        \label{Test PRR at 110 speed, 23 and 26 dBm}
    }
     \caption{$\mathrm{PRR}$ achieved across various $\mathrm{MCS}$ levels and distances. In this simulation, $\rho$ = $50$ \si{vehicles/\kilo\metre}, \(P_{\mathrm{t}} = 23\)~\si{\dBm}, and $\mathrm{SCS}$ = $30$ \si{\kilo\hertz}, (a) $v$ = $50$, (b) $v$ = $80$ and (c) $v$ = $110$ \si{\kilo\metre/\hour}.}
    \label{FIG8}
\end{figure*}

	\begin{figure}[t]
		\centerline{\includegraphics[width=0.5\textwidth]{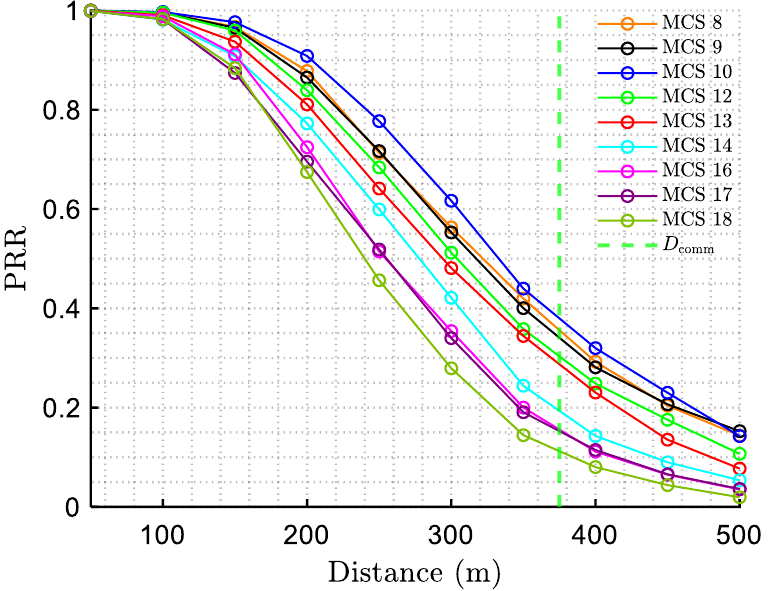}}
		\centering
		\caption{$\mathrm{PRR}$ achieved across various $\mathrm{MCS}$ levels and distances. In this simulation, $\rho$ = $50$ \si{vehicles/\kilo\metre}, $v$ = $50$ \si{\kilo\metre/\hour}, \(P_{\mathrm{t}} = 23\)~\si{\dBm}, and $\mathrm{SCS}$ = $30$ \si{\kilo\hertz}.}
		\label{MCS good1}
	\end{figure}
	
	\subsection{The Impact of Transmit Power on PRR}
	According to Fig.~\ref{23dBm, 30 compare 15 SCS MCS 8}, when \({P}_{\mathrm{t}}=23\) \si{\dBm}, a vehicle density of $30$ vehicles/\si{\kilo\meter} achieves a relatively reliable $\mathrm{PRR}$ of approximately $0.75$, as interference levels are lowest under this condition. Fig.~\ref{24 dBm}, \ref{25 dBm} and \ref{26 dBm} show the $\mathrm{PRR}$ results when \({P}_{\mathrm{t}}\) is set to $24$, $25$, and $26$~\si{\dBm}, respectively. Higher \({P}_{\mathrm{t}}\) decreases the minimum \( D_{\mathrm{comm}} \) and allows greater vehicle densities by improving overall communication quality. The corresponding minimum \( D_{\mathrm{comm}} \) values for each \({P}_{\mathrm{t}}\) are indicated in each figure, and serve as safety thresholds for comparative analysis across the simulation tasks.

	In the case when \({P}_{\mathrm{t}}=24\) \si{\dBm}, the $\mathrm{PRR}$ reaches $0.75$ at a vehicle density of $50$ vehicles/\si{\kilo\meter}. When \({P}_{\mathrm{t}}=25\) \si{\dBm}, the $\mathrm{PRR}$ exceeds $0.9$ and $0.8$ at $80$ and $100$ vehicles/\si{\kilo\meter}, respectively. When \({P}_{\mathrm{t}}=26\) \si{\dBm}, the $\mathrm{PRR}$ remains above $0.95$ across all vehicle densities. These results indicate that while a lower \({P}_{\mathrm{t}}\) of $23$ \si{\dBm} supports only $30$ vehicles/\si{\kilo\meter}, increasing \( {P}_{\mathrm{t}} \) to $26$ \si{\dBm} enables an improved communication and supports up to $100$ vehicles/\si{\kilo\meter}. The impact of \({P}_{\mathrm{t}}\) on $\mathrm{PRR}$ has also been evaluated under different vehicle speeds, as shown in Fig.~\ref{Test PRR at 50 speed, 23 and 26 dBm},~\ref{Test PRR at 80 speed, 23 & 26 dBm} and~\ref{Test PRR at 110 speed, 23 and 26 dBm}. The results illustrated that increasing \({P}_{\mathrm{t}}\) can effectively increase $\mathrm{PRR}$ in all cases. Under the average vehicle speed of $50$, $80$ and $110$ \si{\kilo\metre/\hour}. Increasing \({P}_{\mathrm{t}}\) from $23$ to $26$ \si{dBm} at $50$ vehicles per \si{\kilo\metre}, performs a clear trend in increasing the $\mathrm{PRR}$ for at least $0.2$. Even under heavy traffic of $100$ vehicles per \si{\kilo\metre}, the $\mathrm{PRR}$ increases for a maximum of $0.1$ in all cases. Thus, a higher \({P}_{\mathrm{t}}\) not only improves communication reliability, but also allows significantly higher vehicle densities in rural scenarios.  However, this improvement is achieved at the cost of increased power consumption, which reduces overall energy efficiency, which will be discussed later.
	
	\begin{figure}[t]
		\centerline{\includegraphics[width=0.5\textwidth]{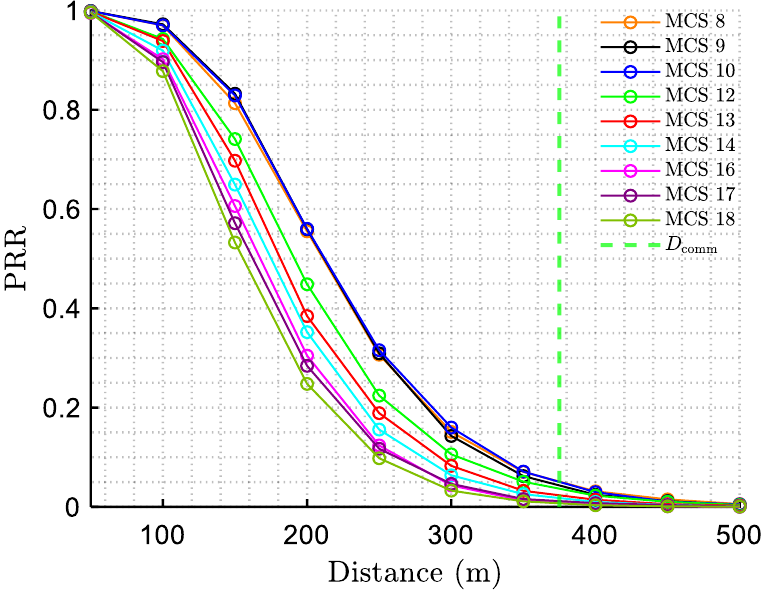}}
		\centering
		\caption{$\mathrm{PRR}$ achieved across various $\mathrm{MCS}$ levels and distances. In this simulation, $\rho$ = $100$ \si{vehicles/\kilo\metre}, $v$ = $50$ \si{\kilo\metre/\hour}, \(P_{\mathrm{t}} = 23\)~\si{\dBm}, and $\mathrm{SCS}$ = $30$ \si{\kilo\hertz}.}
		\label{MCS bad1}
	\end{figure}
	
	\begin{figure}[t]
		\centerline{\includegraphics[width=0.5\textwidth]{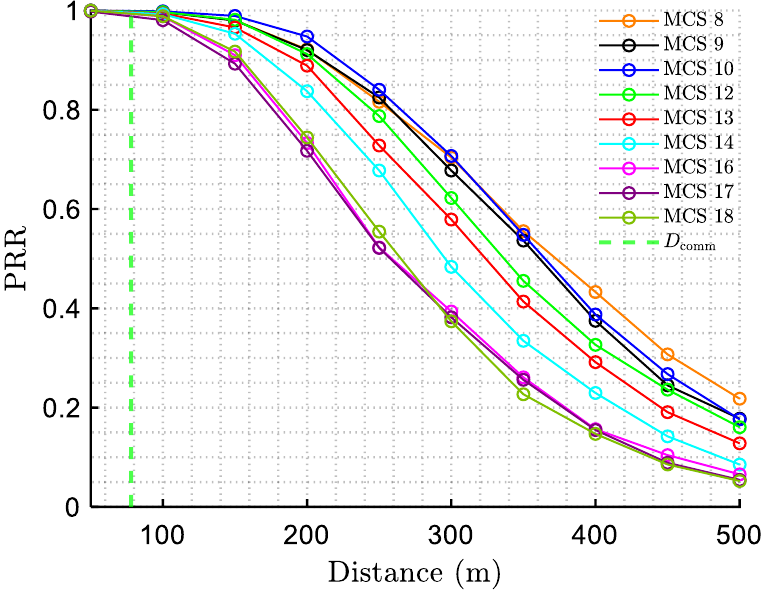}}
		\centering
		\caption{$\mathrm{PRR}$ achieved across various $\mathrm{MCS}$ levels and distances. In this simulation, $\rho$ = $50$ \si{vehicles/\kilo\metre}, $v = 50$ \si{\kilo\metre/\hour}, \(P_{\mathrm{t}} = 26\)~\si{\dBm}, and $\mathrm{SCS} = 30$ \si{\kilo\hertz}.}
		\label{MCS good2}
	\end{figure}
	
	\begin{figure}[t]
		\centerline{\includegraphics[width=0.5\textwidth]{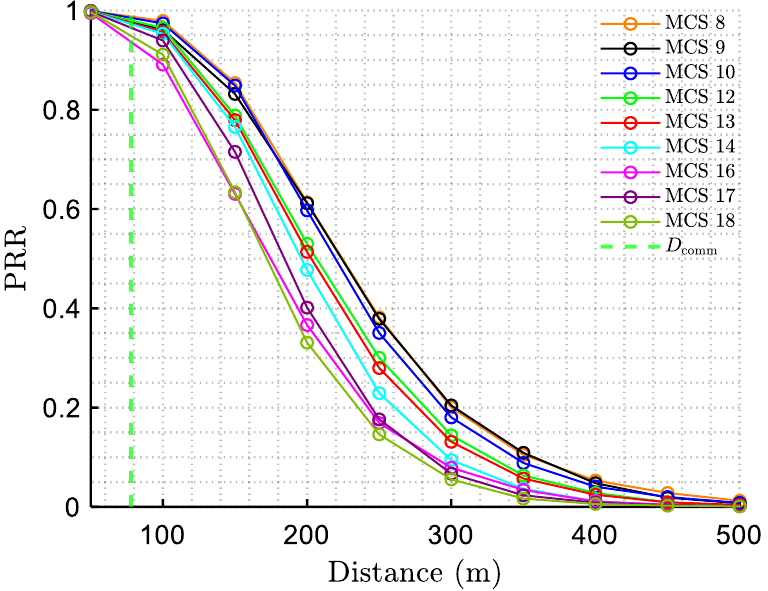}}
		\centering
		\caption{$\mathrm{PRR}$ achieved across various $\mathrm{MCS}$ levels and distances. In this simulation, $\rho$ = $100$ \si{vehicles/\kilo\metre}, $v$ = $50$ \si{\kilo\metre/\hour}, \(P_{\mathrm{t}} = 26\)~\si{\dBm}, and $\mathrm{SCS}$ = $30$ \si{\kilo\hertz}.}
		\label{MCS bad2}
	\end{figure}
	
	\subsection{The Impact of MCS on PRR}
	
	Figs.~\ref{MCS good1}, \ref{MCS bad1}, \ref{MCS good2} and \ref{MCS bad2} present the $\mathrm{PRR}$ results for various $\mathrm{MCS}$ values, transmit powers and traffic densities. The results simulated under light traffic conditions of $50$ \si{vehicles/\kilo\metre} are shown in Figs.~\ref{MCS good1} and \ref{MCS good2}. The results indicate that $\mathrm{PRR}$ declines sharply when $\mathrm{MCS}$ levels exceed $10$, while RF energy consumption increases significantly. The most noticeable improvement occurs between $\mathrm{MCS}$ $8$ and $\mathrm{MCS}$ $10$ for both simulations, where $\mathrm{PRR}$ rises by more than $0.05$ at distances below $300$ \si{\meter}, with additional gains observed between $300$ and $500$ \si{\meter}. The results simulated under heavy traffic conditions of $100$ \si{vehicles/\kilo\metre} are shown in Figs.~\ref{MCS bad1} and \ref{MCS bad2}, at a density of $100$ vehicles/\si{\kilo\meter}, the difference in $\mathrm{PRR}$ between $\mathrm{MCS}$ $8$ and $\mathrm{MCS}$ $10$ becomes negligible for both transmit powers of $23$ and $26$ \si{\dBm}. Furthermore, beyond $350$ \si{\meter}, $\mathrm{MCS}$ $8$ outperforms $\mathrm{MCS}$ $10$ at $26$ \si{\dBm} under lighter traffic conditions with $50$ vehicles/\si{\kilo\meter}. The results highlight three key points. First, while higher $\mathrm{MCS}$ levels increase spectral efficiency, they reduce reliability due to an increase in sensitivity to interference and channel impairments. Second, the limited gains beyond $\mathrm{MCS}$ $10$ are outweighed by the extra retransmissions required in dense traffic, leading to higher energy costs without improving effective $\mathrm{PRR}$. Finally, it is therefore preferable to use a lower MCS, which maintains reliable communication and lower RF energy consumption. Consider higher $\mathrm{MCS}$ levels demand higher \(E_b/N_0\) and are more prone to retransmissions in congested or long-range scenarios, it is essential to select the lowest practical MCS. The lowest $\mathrm{MCS}$ value of $8$, with stable performance in all traffic conditions establishes it as a robust operating point, balancing reliability and energy efficiency at the safety threshold distance \(D_{\mathrm{comm}}\). These findings demonstrate the limited benefit of employing higher $\mathrm{MCS}$ levels, which provide little improvement in reliability while imposing higher energy costs, particularly in dense traffic scenarios.

\begin{figure*}[t]
    \centering
    \subfloat{%
        \includegraphics[width=0.4\textwidth]{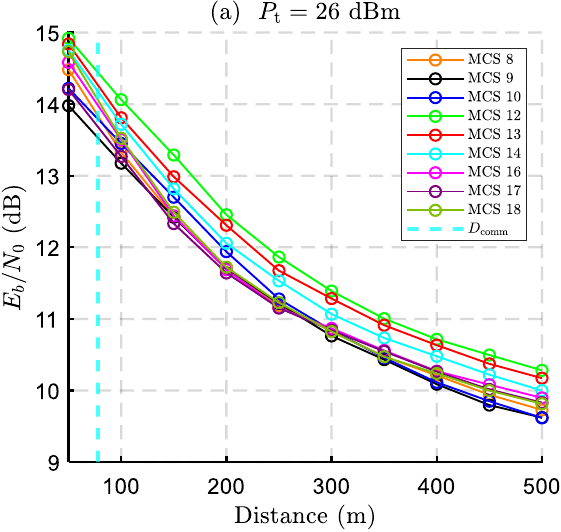}
        \label{fig:15a}
    }
    \hfill
    \subfloat{%
        \includegraphics[width=0.4\textwidth]{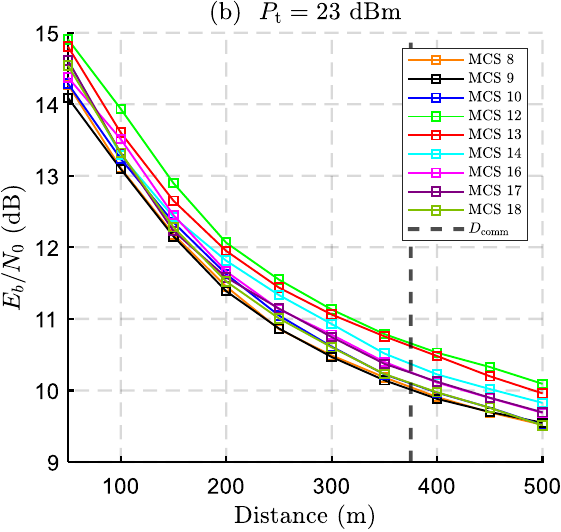}
        \label{fig:15b}
    }
     \caption{\(E_b/N_0\) achieved when (a) \(P_{\mathrm{t}} = 26\)~\si{\dBm}, and (b) \(P_{\mathrm{t}} = 23\)~\si{\dBm}, with an average vehicle speed of $50$ \si{\kilo\metre/\hour} with $\rho$ = $50$ \si{vehicles/\kilo\metre}, under multiple $\mathrm{MCS}$ configurations.}
    \label{23 - 26 dBm Rho 50 EbN0 multi mcs}
\end{figure*}   
    \begin{figure*}[t]
    \centering
    \subfloat{%
        \includegraphics[width=0.4\textwidth]{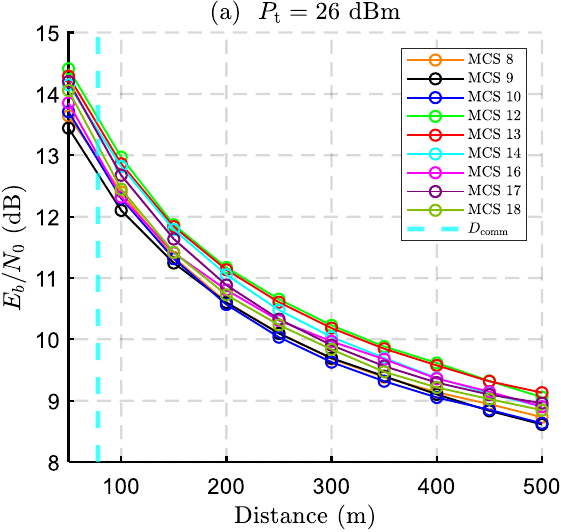}
        \label{fig:16a}
    }
    \hfill
    \subfloat{%
        \includegraphics[width=0.4\textwidth]{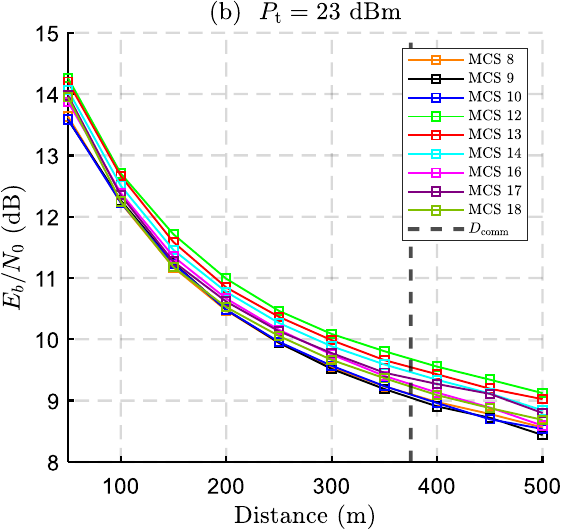}
        \label{fig:16b}
    }
     \caption{\(E_b/N_0\) achieved when (a) \(P_{\mathrm{t}} = 26\)~\si{\dBm}, and (b) \(P_{\mathrm{t}} = 23\)~\si{\dBm}, with an average vehicle speed of $50$ \si{\kilo\metre/\hour} with $\rho$ = $100$ \si{vehicles/\kilo\metre}, under multiple $\mathrm{MCS}$ configurations.}
    \label{23 - 26 dBm Rho 100 EbN0 multi mcs}
\end{figure*}   
	\begin{figure*}[t]
		\centerline{\includegraphics[width=\textwidth]{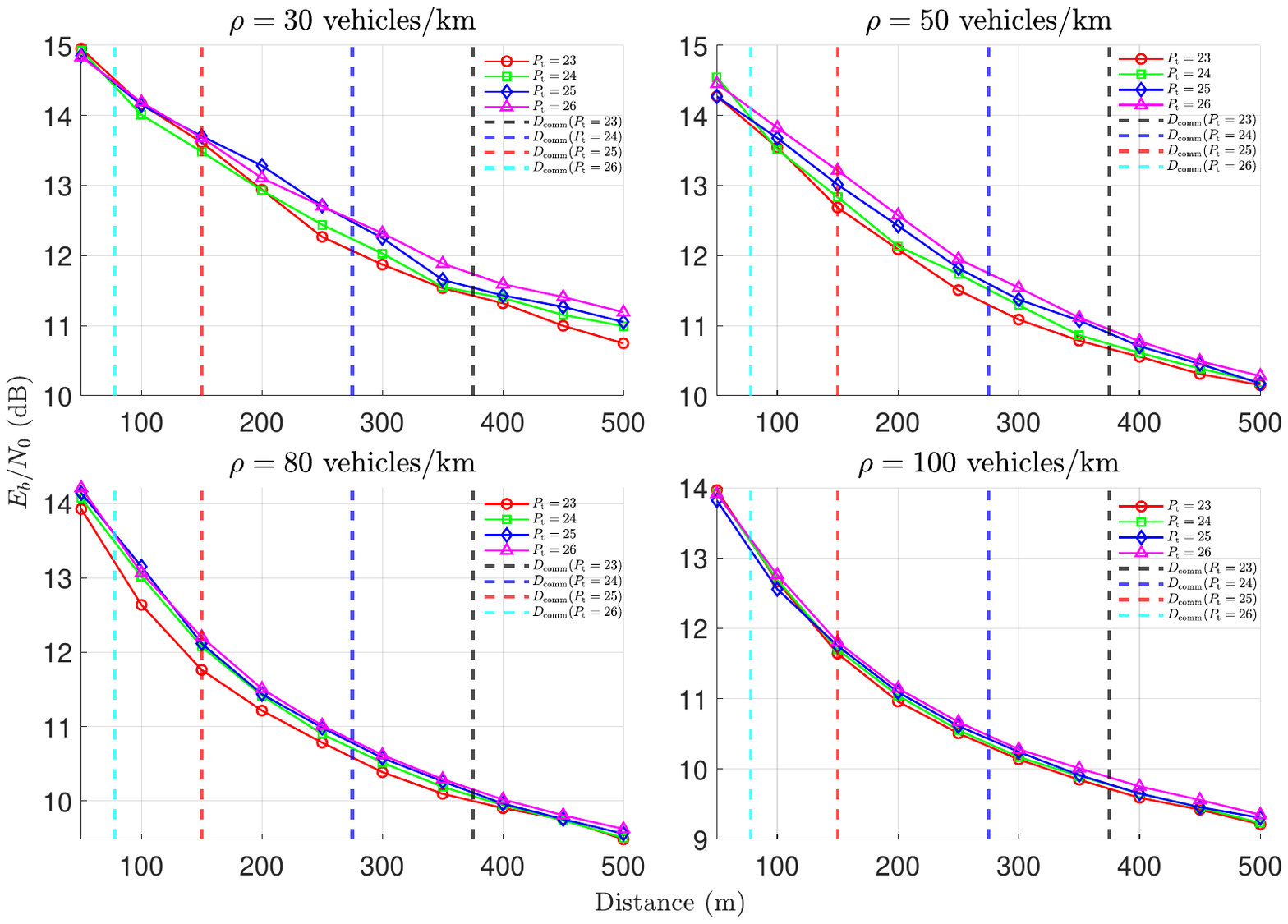}}
		\centering
		\caption{\(E_b/N_0\) achieved when \(P_{\mathrm{t}}\) ranging from $23$ to $26$~\si{\dBm}.}
		\label{23 - 26 dBm EbN0}
	\end{figure*}

\begin{figure}[h]
\centerline{\includegraphics[width=0.5\textwidth]{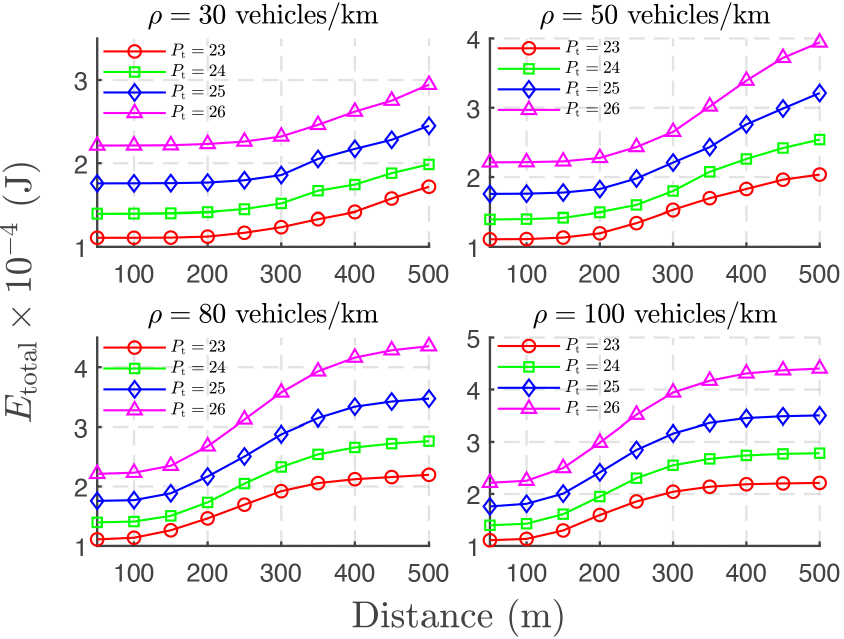}}
\centering
\caption{Total energy spent on a single transmission with truncated HARQ, expressed in $10^{-4}$ \si{J}.}
\label{Etotal}
\end{figure}

\begin{table}[t]
    \centering
    \caption{Total number of transmissions from RSU to all vehicles at \(\mathrm{SCS}=30\), \(\mathrm{MCS}=8\), and $v$ = $50$ \si{\kilo\metre/\hour}}. 
    \label{tab:transmissions}
    \renewcommand{\arraystretch}{1.25}
    \setlength{\tabcolsep}{6pt}
    \begin{tabular}{|c|c|c|c|c|}
        \hline
        \textbf{\(P_{\mathrm{t}}\) (\si{\dBm})} & \textbf{\(\rho=30\)} & \textbf{\(\rho=50\)} & \textbf{\(\rho=80\)} & \textbf{\(\rho=100\)} \\ 
        \hline 
        \hline
        $23$ & $528209$ & $1476832$ & $3721318$ & $5677773$ \\ \hline
        $24$ & $518580$ & $1391392$ & $3630052$ & $5755369$ \\ \hline
        $25$ & $495324$ & $1401655$ & $3683343$ & $5797117$ \\ \hline
        $26$ & $535225$ & $1469037$ & $3665039$ & $5731340$ \\ \hline
    \end{tabular}
\end{table}


\begin{table}[t]
    \centering
    \caption{Maximum RF energy cost ($J$) for delivering packets from \texttt{RSU} to all vehicles at \(\mathrm{SCS}=30\), \(\mathrm{MCS}=8\), and $v$ = $50$ \si{\kilo\metre/\hour}.}
    \label{tab:total_energy_cost}
    \renewcommand{\arraystretch}{1.25}
    \setlength{\tabcolsep}{6pt}
    \begin{tabular}{|c|c|c|c|c|}
        \hline
        \textbf{\(P_{\mathrm{t}}\) (\si{\dBm})} & \textbf{\(\rho=30\)} & \textbf{\(\rho=50\)} & \textbf{\(\rho=80\)} & \textbf{\(\rho=100\)} \\ 
        \hline
        \hline
        $23$ & $90.69$ & $300.68$ & $817.57$ & $1255.92$ \\ \hline
        $24$ & $102.94$ & $353.55$ & $1003.35$ & $1602.30$ \\ \hline
        $25$ & $121.31$ & $450.21$ & $1278.49$ & $2029.57$ \\ \hline
        $26$ & $157.52$ & $579.09$ & $1795.39$ & $2523.51$ \\ \hline
    \end{tabular}
\end{table}

\begin{table}[t]
    \centering
    \caption{Maximum communication energy cost ($J$) for delivering packets from \texttt{RSU} to all vehicles at \(\mathrm{SCS}=30\), \(\mathrm{MCS}=8\), and $v$ = $50$ \si{\kilo\metre/\hour}.}
    \label{tab:total_comms_energy_cost}
    \renewcommand{\arraystretch}{1.25}
    \setlength{\tabcolsep}{6pt}
    \begin{tabular}{|c|c|c|c|c|}
        \hline
        \textbf{\(P_{\mathrm{t}}\) (\si{\dBm})} & \textbf{\(\rho=30\)} & \textbf{\(\rho=50\)} & \textbf{\(\rho=80\)} & \textbf{\(\rho=100\)} \\ 
        \hline
        \hline
        $23$ & $140.52$ & $455.75$ & $1231.70$ & $1889.75$ \\ \hline
        $24$ & $148.30$ & $498.68$ & $1407.17$ & $2244.56$ \\ \hline
        $25$ & $164.05$ & $596.95$ & $1687.16$ & $2675.75$ \\ \hline
        $26$ & $201.46$ & $728.93$ & $2250.75$ & $3161.76$ \\ \hline
    \end{tabular}
\end{table}

	\section{Analysis on Energy Consumption}
    \label{se:analysis}
	This section provides the analysis of RF energy consumption, which enables us to evaluate the impact of $\mathrm{MCS}$ configuration and identify the most suitable code rate. The target is to achieve a good balance between $\mathrm{PRR}$ performance and energy efficiency under both light and heavy traffic conditions.
	
	\subsection{The Impact of MCS on Achieved \(E_b/N_0\)}
    Fig.~\ref{23 - 26 dBm Rho 50 EbN0 multi mcs} and Fig.~\ref{23 - 26 dBm Rho 100 EbN0 multi mcs} show the obtained \(E_b/N_0\) under various $\mathrm{MCS}$ configurations the light and heavy traffic, respectively. The effect of \({P}_{\mathrm{t}}\), ranging from $23$ to $26$ \si{\dBm}, is examined for vehicle densities of $50$ and $100$ vehicles per \si{\kilo\metre}, revealing trade-offs between energy efficiency and communication reliability. As can be seen from Fig.~\ref{23 - 26 dBm Rho 50 EbN0 multi mcs}, under light traffic of $50$ vehicles per \si{\kilo\meter}, increasing \( {P}_{\mathrm{t}} \) is expected to make significant improvements across all $\mathrm{MCS}$ values. However, noticeable gains are only observed at extremely high $\mathrm{MCS}$ levels, $16$ to $18$, which is not energy efficient. For lower $\mathrm{MCS}$ values of $8$ to $10$, the system maintains consistently robust performance, with \(E_b/N_0\) values remaining above $10$~\si{\dB} throughout the range of \(D_{\mathrm{comm}}\), indicating reliable communication quality under the tested conditions. As shown in  Fig.~\ref{fig:15a}, within the \( D_{\text{comm}} \), $\mathrm{MCS}$ $8$ achieves an \(E_b/N_0\) of $10.2$ \si{\dB} at $23$ \si{\dBm} and $10.5$ \si{\dB} at $26$ \si{\dBm}, indicating only a marginal improvement. Compare with Fig.~\ref{fig:15b}, these findings suggest that the use of $\mathrm{MCS}$ $8$ to $10$ with the setting \( {P}_{\mathrm{t}} = 23\) \si{\dBm} is energy efficient and sufficient for light traffic scenarios. In contrast, the worst case is shown in Fig.~\ref{23 - 26 dBm Rho 100 EbN0 multi mcs} which exhibits a tighter clustering of all $\mathrm{MCS}$ curves and a reduced average \(E_b/N_0\). This indicates higher interference and reduced link quality. Moreover, compare between the results shown in Fig.~\ref{fig:16a} and \ref{fig:16b}, the performance gap between $26$ and $23$ \si{\dBm} reduces, and the benefit of using higher $\mathrm{MCS}$ becomes less pronounced. Under such conditions, lower $\mathrm{MCS}$ levels, such as $8$ to $10$ become more resilient and energy efficient. Higher $\mathrm{MCS}$ levels are struggling to maintain sufficient \(E_b/N_0\) beyond $300$ m, which results in increased retransmissions and higher RF energy consumption. Finally, based on the analysis under both light and heavy traffic conditions, a lower level $\mathrm{MCS}$ with a value of $8$, which demonstrates consistently good performance across scenarios. Therefore, $\mathrm{MCS}$ $8$ is selected to analyse the trade-offs between RF energy consumption and safety in the next subsections.

	\subsection{The Impact of $P_t$ on Achieved \(E_b/N_0\)} Fig.~\ref{23 - 26 dBm EbN0}  shows the achieved \(E_b/N_0\) against the distance between \texttt{RSU} and \texttt{Ego}, under various \({P}_{\mathrm{t}}\) values of $23$, $24$, $25$, and $26$ \si{\dBm} and vehicle densities of $30$, $50$, $80$, and $100$ vehicles per \si{\kilo\metre}. The corresponding energy expenditure, measured as the average number of transmissions required to successfully deliver $350$ packets, is summarised in Table~\ref{tab:transmissions}. These results provide insights on the configuration of \({P}_{\mathrm{t}}\) and $\mathrm{MCS}$ to minimise RF energy consumption while maintaining reliable V2X connectivity in ruralenvironments. According to Fig.~\ref{23 - 26 dBm EbN0}, the results reveal a clear trade-off between energy efficiency and communication reliability across vehicle density scenarios. At low densities (e.g., $30$ to $50$ vehicles per \si{\kilo\meter}), reducing \({P}_{\mathrm{t}}\) from $26$ to $23$ \si{\dBm} leads to only minor reductions in \(E_b/N_0\), particularly in the short to medium ranges. For example, the loss is approximately $1.5$ \si{\dB} at $100$ \si{\meter} and remains within a tolerable margin at $200$ \si{\meter} for stable packet reception. These findings indicate that employing a lower \({P}_{\mathrm{t}}\), such as $23$ \si{\dBm}, is sufficient in low-density traffic to conserve energy without compromising communication reliability, thus improving the efficiency of V2X systems in rural scenarios.
	
	In contrast, an increase in \({P}_{\mathrm{t}}\) from $23$ to $26$ \si{\dBm} provides substantial performance gains under high vehicle densities ($80$ to $100$ vehicles per \si{\kilo\meter}), mainly by extending the \(D_{\mathrm{comm}}\) needed for stable connected braking. At $250$ \si{\meter}, the increase in \(E_b/N_0\) is about $3$ \si{\dB}, which improves packet reception reliability and limits interference-induced retransmissions. These results demonstrate that higher \({P}_{\mathrm{t}}\), such as $26$ \si{\dBm}, is essential in congested traffic to overcome channel degradation, reduce packet loss, and preserve the safety margin required for V2X safety-critical operations.

\subsection{Analysis on Energy Consumption} 
   
   Fig.~\ref{Etotal} illustrates the total transmission energy required to deliver all packets from the \texttt{RSU} to the \texttt{Ego}, \(E_{\mathrm{total}}\) against the distance, under various \({P}_{\mathrm{t}}\) values of $23$, $24$, $25$, and $26$ \si{\dBm} and vehicle densities of $30$, $50$, $80$, and $100$ vehicles per \si{\kilo\metre}. In all traffic scenarios, \(E_{\mathrm{total}}\) increases monotonically with distance, consistent with reduced $\mathrm{PRR}$ and shrinking link margin at longer ranges. For any fixed vehicle density and distance, a higher \(P_{\mathrm{t}}\) yields higher energy per packet, while \(P_{\mathrm{t}}=23~\mathrm{dBm}\) is the most energy efficient. Under an increased vehicle density such as $80$ and $100$ vehicles per \si{\kilo\metre}, the channel is busier with increased interference, which shifts the curves upward and steepens their slope beyond $200$ and $300$ \si{m}. This effect is most evident for heavy traffic, such as $80$ and $100$ per \si{\kilo\metre}, where the gap between low and high \(P_{\mathrm{t}}\) widens with distance. These trends support density and distance-aware power control, under lighter traffic of $30$ and $50$ vehicles per \si{\kilo\metre}, \(P_{\mathrm{t}}=23\) to \(24~\mathrm{dBm}\) is adequate for several hundred metres, under heavier traffic of $80$ and $100$ vehicles per \si{\kilo\metre}, increasing \(P_{\mathrm{t}}\) should be reserved for operation near the respective communication limit to meet reliability while avoiding unnecessary energy expenditure.
    
Table~\ref{tab:transmissions} illustrates the averaged total generated transmissions from \texttt{RSU} to all the vehicles, based on increasing \({P}_{\mathrm{t}}\). As shown in Table~\ref{tab:transmissions}, the total number of transmissions does not strictly decrease monotonically with \(P_{\mathrm{t}}\). This behaviour is due to randomness in the simulation, caused by shadowing fading, traffic mobility, and random scheduling of retransmissions, all of which introduce small fluctuations in aggregate transmission counts. However, a general decreasing trend can still be observed in the results. Under light traffic conditions ($30$ and $50$ vehicles per \si{\kilo\meter}), increasing \(P_{\mathrm{t}}\) from $23$ to $25$ \si{dBm} decreased the total number of transmissions for about $30000$. A similar pattern is observed under heavy traffic ($30$ and $50$ vehicles per \si{\kilo\meter}), where increasing \(P_{\mathrm{t}}\) from $25$ to \(26~\mathrm{dBm}\) provides moderate reliability improvements while reducing the retransmission overhead.

    Table~\ref{tab:total_energy_cost} illustrates the maximum RF energy for delivering all packets from \texttt{RSU} to all vehicles. Under light and moderate traffic densities of $30$ and $50$ vehicles per \si{\kilo\meter}, the increase in \({P}_{\mathrm{t}}\) from $23$ to $26$~\si{\dBm} results in a more gradual increase in total energy from \(90.6\) to \(157.52~\mathrm{J}\) and from \(300.68\) to \(579.09~\mathrm{J}\), representing approximately a $74$\% and $93$\% rise, respectively. This shows that in sparse environments, higher \({P}_{\mathrm{t}}\) can improve coverage and $\mathrm{PRR}$ with only a moderate increase in RF energy consumption. However, once traffic density exceeds a critical threshold (e.g. $80$ vehicles per \si{\kilo\meter}), the marginal energy gain per successfully delivered packet decreases sharply. In contrast to light traffic conditions, where moderate increases in \({P}_{\mathrm{t}}\) yield energy-efficient gains, heavy traffic scenarios exhibit a different trade-off. At $\rho = 100$, raising \({P}_{\mathrm{t}}\) from $23$ to $26$~\si{\dBm} increases the total number of transmissions and consequently leads to a substantial increase in total RF energy consumption, from about \(1255.92~\mathrm{J}\) to \(2523.51~\mathrm{J}\), the growth in total energy cost is much more pronounced than light traffics, approximately $101$\%. This steep growth indicates that in dense vehicular networks, additional \({P}_{\mathrm{t}}\) has proportionally improved communication reliability, but also amplified the overall energy cost to deliver all packets by \texttt{RSU}. Table~\ref{tab:total_comms_energy_cost} shows the total communication energy expense with various values of \({P}_{\mathrm{t}}\). The results show that RF energy still contributes more than $60\%$ of the overall communication energy across the considered configurations. This indicates that RF energy consumption remains a dominant and controllable component, and therefore motivates focusing on selecting appropriate communication control factors to reduce marginal RF energy while maintaining reliability in rural V2X scenarios.
    
   The results highlight the trade-off between communication reliability and energy efficiency in V2X networks. Increasing \({P}_{\mathrm{t}}\) effectively extends \(D_{\mathrm{comm}}\) and enhances the PRR, thus improving transmission reliability and road safety. However, this improvement shows inconsistency, which can be attributed to random fluctuations of shadowing and traffic layouts during simulations. Consequently,  total RF energy consumption increases sharply with \({P}_{\mathrm{t}}\), in light traffic scenarios, a lower \({P}_{\mathrm{t}}\) of $23$ \si{\dBm} is sufficient to maintain effective communication while reducing energy costs. However, in high-density scenarios, a higher \({P}_{\mathrm{t}}\) of $26$ \si{\dBm} is effective in extending the scenario \( D_{\mathrm{comm}} \) to improve safety performance, but this cost of energy from excessive power needs to be considered, as it does not effectively reduce the total number of transmissions. These findings suggest that, while increasing \({P}_{\mathrm{t}}\) improves connectivity and safety performance, it does so at the expense of energy efficiency. Hence, a balance must be achieved between reliable communication and sustainable energy use. Therefore, an adaptive power control mechanism that dynamically adjusts \({P}_{\mathrm{t}}\) according to vehicle density is essential to improve both safety performance and overall RF energy efficiency in future rural V2X systems.

\section{Conclusions}
\label{se:conclusion}

 This study investigates the impact of three key communication control factors: $\mathrm{SCS}$, \({P}_{\mathrm{t}}\), and $\mathrm{MCS}$ on safety-relevant performance metrics and RF energy consumption for rural connected braking. We evaluated $\mathrm{PRR}$ and \(D_{\mathrm{comm}}\) under diverse traffic conditions, followed by an energy consumption analysis. The results demonstrate that, an $\mathrm{SCS}$ of $30$ \si{\kilo\hertz} combined with $\mathrm{MCS}$ 8 consistently ensures reliable $\mathrm{PRR}$ at the safety threshold distance \(D_{\mathrm{comm}}\). Under light traffic, where interference is limited, maintaining a lower \({P}_{\mathrm{t}}\) is more energy efficient while still preserving safety margins at \(D_{\mathrm{comm}}\). Under heavy traffic, increasing \({P}_{\mathrm{t}}\) from $23$ \si{\dBm} to $26$ \si{\dBm} can effectively improve safety by increasing $\mathrm{PRR}$ and \(D_{\mathrm{comm}}\). However, this improvement comes at the cost of more than a $100$\% increase in total energy expenditure to deliver all packets from the \texttt{RSU} to every vehicle, although it provides substantial gains in communication reliability and enables support for higher vehicle densities. These findings underscore the importance of adaptive control of three communication factors to balance the safety and energy consumption of rural connected brakes. Overall, the findings confirm that adaptive configuration of \({P}_{\mathrm{t}}\), $\mathrm{SCS}$, and $\mathrm{MCS}$ is essential to balance energy efficiency with safety-critical performance in rural V2X systems. By explicitly incorporating \(D_{\mathrm{comm}}\) as a design metric, this work provides a detail analysis on performance evaluation of communication control factors to ensuring a safe, energy-efficient operation of connected braking applications in rural traffic conditions, in order to bring useful insights to the future rural CAV developments. A future research direction is the development of adaptive configuration mechanisms that dynamically adjust communication control factors according to real-time traffic and channel conditions in rural environments.

	\section*{Acknowledgements}
	For the purpose of open access, the author has applied a Creative Commons Attribution (CC BY) license to any Author Accepted Manuscript version arising from this submission.
	
	\bibliographystyle{IEEEtran}
	\bibliography{references}

\end{document}